\documentclass{amsart}
 \usepackage{amsaddr}
\usepackage[english]{babel}

\usepackage[letterpaper,top=2cm,bottom=2cm,left=3cm,right=3cm,marginparwidth=1.75cm]{geometry}

\usepackage[colorlinks=true, allcolors=blue]{hyperref}

\usepackage{graphicx}%
\usepackage{multirow}%
\usepackage{amsmath,amssymb,amsfonts}%
\usepackage{amsthm}%
\usepackage{mathtools}
\usepackage{mathrsfs}%
\usepackage{xcolor}%
\usepackage{textcomp}%
\usepackage{manyfoot}%
\usepackage{booktabs}%
\usepackage{algorithm}%
\usepackage{algorithmicx}%
\usepackage{algpseudocode}%
\usepackage{listings}%

\newtheorem{thm}{Theorem} 
\newtheorem{prop}[thm]{Proposition} 
\newtheorem{lem}[thm]{Lemma} 
 
\DeclareMathOperator{\indeg}{indeg}
\DeclareMathOperator{\outdeg}{outdeg}
\DeclareMathOperator{\CRS}{CRS}
\DeclareMathOperator{\MCRS}{MCRS}
\DeclareMathOperator{\RP}{RP}
\DeclareMathOperator{\MRP}{MRP}
\DeclareMathOperator{\ARP}{ARP}
\DeclareMathOperator{\Supp}{Supp}
\newcommand{\cT}{\mathcal{T}}
\newcommand{\cS}{\mathcal{S}}
\newcommand{\tC}{\text{\upshape\sffamily C}}
\newcommand{\tR}{\text{\upshape\sffamily R}}
\newcommand{\tN}{\text{\upshape\sffamily N}}
\newcommand{\tP}{\text{\upshape\sffamily P}}
\newcommand{\tS}{\text{\upshape\sffamily S}}
\newcommand{\tT}{\text{\upshape\sffamily T}}

\newcommand{\expair}[3]{\prescript{(#1,#2)}{}{#3}}
\newcommand{\aug}[2]{\prescript{#1}{}{#2}}

\title{Generation of orchard and tree-child networks} 

\author{Gabriel Cardona}
\address{Department of Mathematics and Computer Science,  University of the Balearic Islands, E-07122 Palma, Spain}
\author{Gerard Ribas}
\address{Higher Polytechnic School,  University of the Balearic Islands, E-07122 Palma, Spain}

\author{Joan Carles Pons}
\address{Department of Mathematics and Computer Science,  University of the Balearic Islands, E-07122 Palma, Spain}
\email{joancarles.pons@uib.es}

\begin{document}
\maketitle

\begin{abstract}  Phylogenetic networks are an extension of phylogenetic trees that allow for the representation of reticulate evolution events. One of the classes of networks that has gained the attention of the scientific community over the last years is the class of orchard networks, that generalizes tree-child networks, one of the most studied classes of networks.

  In this paper we focus on the combinatorial and algorithmic problem of the generation of orchard networks, and also of tree-child networks. To this end, we use that these networks are defined as those that can be recovered by a reversing a certain reduction process. Then, we show how to choose a ``minimum'' reduction process among all that can be applied to a network, and hence we get a unique representation of the network that, in fact, can be given in terms of sequences of pairs of integers, whose length is related to the number of leaves and reticulations of the network. 
  Therefore, the generation of networks is reduced to the generation of such sequences of pairs. Our main result is a recursive method for the efficient generation of all minimum sequences, and hence of all orchard (or tree-child) networks with a given number of leaves and reticulations.

An implementation in C of the algorithms described in this paper, along with some computational experiments, can be downloaded from the public repository~\url{https://github.com/gerardet46/OrchardGenerator}. Using this implementation, we have computed the number of orchard networks with at most $6$ leaves and $8$ reticulations.
\end{abstract}

\section{Introduction}

During decades, phylogenetic trees have been the model used to represent the  branching pattern for the evolution of a set of Operational Taxonomic Units (OTUs for short). From the 1980s onward, it became evident that phylogenetic networks were a more accurate framework, with the potential to cover more complex evolutionary scenarios such as hybridizations, recombinations, or lateral gene transfers. 

In the broadest sense, phylogenetic networks are directed acyclic graphs 
whose leaves are labelled by the organisms under study. 
This general definition, while allowing a wide range of biological processes to be considered, lacks mathematical tractability. For this reason, some other constraints must be considered, resulting in a wide variety of classes of phylogenetic networks (see \cite{kong2022classes} for a recent review, or 
  \cite[Chapter~10]{steel2016phylogeny}). In this work we focus on the class of \emph{orchard} networks~\cite{erdos2019class} (also called \emph{cherry-picking} networks \cite{Janssen2021}) and \emph{tree-child} networks~\cite{Cardona2009}, a subclass of the first and one of the most explored classes of networks. 

 Orchard networks are networks that can be reduced to a trivial network by iteratively identifying and reducing certain substructures (namely, cherries and reticulated cherries) involving two leaves. Orchard networks are one of those classes of networks with biological significance (according to \cite{kong2022classes}) because they can be viewed as a backbone tree with additional ``horizontal'' arcs (see \cite{van2022orchard} for more details).

One of the relevant problems in the study of phylogenetic networks is that of  their sequential generation; that is, obtaining a method to generate them in an efficient and unique way. 
Generation of phylogenetic networks is useful, for example, for testing the performance of methods in phylogenetics and for testing hypotheses about the evolutionary relationships among organisms by the comparison of different network topologies.

Up to our knowledge, there exists no prior work neither on the systematic generation of orchard networks nor on its counting.
Notice, however, that the identification of orchard networks  as trees with extra arcs used in \cite{van2022orchard}, obviously results in an algorithm to generate them, but not uniquely, and moreover there is no  prior control on the probability distribution of the generated networks. The situation for tree-child networks is slightly better, since there are previous works on the enumeration~\cite{fuchs2020counting, pons2021combinatorial} and generation~\cite{cardona2019generation, cardona2020counting} of this kind of networks, but much less efficiently than the method given here (see Section~\ref{sec:computational}).

In this paper, we shall focus on the problem of the effective and injective generation of orchard and tree-child phylogenetic networks; that is, no pair of generated networks will be the same (technically, isomorphic), and we can promptly get many networks with the number of leaves and reticulations that we want.
Our method of generation is based on the construction of sequences of pairs of integers that encode orchard (and, in particular, tree-child) networks as introduced in \cite{Janssen2021}. However, there are different sequences that generate the same network, so that we choose among them a \emph{minimum} one that uniquely represents it. Hence, our strategy to generate orchard (and tree-child) networks is based on the generation of those minimum sequences.

The paper is organized as follows. In Section~\ref{sec:preliminaries} we give basic definitions used throughout the manuscript. In Section~\ref{sec:orchard} we define orchard networks and how they can be reduced by means of reducible sequences. In Section~\ref{sec:min_red_seq} we show that we can choose a minimum (in a sense to be defined) reducible sequence in order to uniquely identify an orchard network up to isomorphism. 
Section~\ref{sec:augmentation} shows how the reduction of a pair can be reverted by means of augmentations, and in Section~\ref{sec:aug_seq_generation} it is used to describe how
to recover an orchard network by reversing the whole reduction process, and how this process, together with the unicity of the minimum reducible sequence, allows us to generate orchard networks injectively. In Section~\ref{sec:generation_tc} we adapt our methods to generate tree-child networks, which constitute a relevant subclass of orchard networks. In Section~\ref{sec:computational} we present the implementation we have made of the methods contained in this paper and exhibit some computational experiments we have performed, including the computation of the number of orchard networks with up to $6$ leaves and $8$ reticulations. Finally, Section~\ref{sec:conclusions} contains the conclusions of the manuscript and some possible directions of future work.

\section{Preliminaries}\label{sec:preliminaries}

Throughout the paper, for any positive integer $n$, we denote by $[n]$ the set $\{1,\dots,n\}$.

The graphs $N=(V,A)$ we shall work with are directed and acyclic. 
Given two nodes $u,v\in V$, if there is an arc with tail $u$ and head $v$ (or from $u$ to $v$), we denote it as $uv$. In that case, $u$ is a \emph{parent} of $v$ and $v$ is a \emph{child} of $u$.

Given a node $u\in V$,  $\indeg u$ (resp. $\outdeg u$) denotes the number of arcs whose head (resp. tail) is $u$. 
We say that  $u$  is \emph{elementary} if $\indeg u=\outdeg u=1$, and its \emph{simplification} consists in removing it (together with its incident arcs) and connecting its single parent to its single child.

Given a set $X$ of \emph{taxa}, a (rooted binary) \emph{phylogenetic network}, or simply a \emph{network}, on $X$,  is a directed acyclic graph $(V,A)$ without parallel arcs such that any node $u\in V$ is either: 
  \begin{enumerate}
        \item a \emph{root}, with $\indeg u=0$, $\outdeg u=1$ (and there can only be one root), or
        \item a \emph{leaf}, with $\indeg u=1$, $\outdeg u=0$, or
        \item a \emph{tree node}, with $\indeg u=1$, $\outdeg u=2$, or
        \item a \emph{reticulation}, with $\indeg u=2$, $\outdeg u=1$,
  \end{enumerate}
together with a fixed  bijection between $X$ and the set of leaves. 

We shall hereafter identify the set $X$ of taxa and the set of leaves, and we shall always assume that $X$ is formed by positive integers, and hence $X\subseteq [n]$ for some $n$.

Two networks $N$ and $N'$ are \emph{isomorphic} if there exists a bijection $\phi$ between the respective set of nodes that reflects and preserves the arcs (that is, $uv$ is an arc in $N$ if, and only if, $\phi(u)\phi(v)$ is an arc in $N'$), which is the identity on the leaves (that is, if $l$ is a leaf, $\phi(l)=l$). Hereafter, we shall simply say that two networks are equal if they are isomorphic.

In case that $X=\{l\}$, for some $l\in[n]$, we define the \emph{trivial} network on $\{l\}$, and denote it by $I_l$, as the network that has two nodes, the root and the leaf $l$, connected by an arc.

\section{Orchard networks}
\label{sec:orchard}

Let $N=(V,A)$ be a network on $X\subseteq[n]$ and let $(i,j)\in X\times X$ with $i\neq j$. Also, denote by $p_i, p_j$ the parents of the leaves $i$ and $j$ in $N$, respectively. We call $(i,j)$ a \emph{cherry} if $p_i=p_j$, and we call it a \emph{reticulated-cherry} if $p_i$ is a reticulation, $p_j$ is a tree node, and $p_j$ is one of the parents of $p_i$. In either case $(i,j)$ is a cherry or a reticulated-cherry, we say that $(i,j)$ is a \emph{reducible pair} in $N$. In order to identify which kind of reducible pair is $(i,j)$ in $N$, we will define its \emph{character} as $\chi_N(i,j)=\tC$ if it is a cherry and $\chi_N(i,j)=\tR$ if it is a reticulated-cherry. 
Notice that the conditions of being a cherry and a reticulated-cherry are clearly incompatible, which implies that $\chi_N$ is well defined.
If the network is clear from the context, we will simply write $\chi(i,j)$. To ease notations, if a pair $(i,j)$ has character $\chi=\chi_N(i,j)$, we shall write the \emph{annotated pair} as $(i,j)^\chi$.

Given a network $N$, we shall denote by $\RP(N)$ the set of reducible pairs of $N$, by $\chi_N$ the mapping $\RP(N)\to\{\tC,\tR\}$ that gives the character of the reducible pairs, and by $\ARP(N)$ the set of annotated reducible pairs of $N$. 

If $(i,j)\in\RP(N)$, the \emph{reduction} of $(i,j)$ in $N$, denoted by $N^{(i,j)}$, is the result of:
\begin{itemize}
    \item If $\chi(i,j)=\tC$, then 
    remove the leaf $i$ (and its incoming arc) and simplify $p_i$, which is now an elementary node.
     \item If $\chi(i,j)=\tR$, then remove the arc $p_jp_i$ and then simplify $p_i$ and $p_j$, which are now elementary nodes. 
\end{itemize}

Given a sequence of pairs of integers $S=(s_1,\dots,s_k)$ which, for brevity, we will write as $S=s_1\cdots s_k$, with $s_t=(i_t,j_t)$ and $i_t,j_t\in[n]$, of length $k\ge1$, we say that $S$ is \emph{reducible} in $N$ if:
\begin{itemize}
    \item $s_1$ is reducible in  $N$.
    \item For every $t\in\{ 2,\dots,k \}$, $s_t$ is reducible in
    $(\dots(N^{s_1})^{s_2}\dots)^{s_{t-1}}$.
\end{itemize}
In such a case, we shall define the \emph{reduction} of $N$ with respect to $S$ as $(\dots(N^{s_1})^{s_2}\dots)^{s_{k}}$ and it will be denoted by $N^S$.

Moreover, we say that $S$ is \emph{complete} if $N^S=I_l$ for some $l\in X$ and, in case one such complete sequence exists, we call $N$ an \emph{orchard network} \cite{erdos2019class, Janssen2021}. We shall also consider the trivial networks $I_l$ as orchard networks, corresponding to the case when the sequence $S$ is empty. Notice that trivial networks are the only ones that have a single leaf.

The fundamental result that allows one to classify orchard networks using complete reducible sequences is the following, which is adapted from~\cite[Corollary~1]{Janssen2021}.

\begin{thm}
    \label{thm:same-crs-implies-iso}
    Let $S$ be a complete reducible sequence for two orchard networks $N$ and $N'$. Then, $N\cong N'$. 
\end{thm}

Notice, however, that the complete reducible sequence for an orchard network is not unique. For instance, Fig.~\ref{fig:reduct} shows an orchard network $N$ together with the networks that are obtained by application of the reductions in the sequence 
  $S=(3,1)(3,2)(1,2)(3,4)(2,4)$, but it is easy to check that $S'= (3,1)(3,4)(2,3)(1,3)(3,4)$ is another complete reducible sequence for $N$.

\begin{figure}[!ht]%
\centering
\includegraphics[width=1\textwidth]{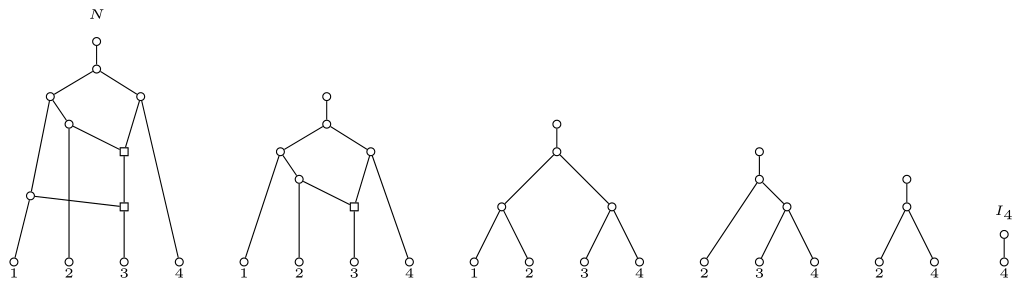}
\caption{ An orchard network $N$ and the set of intermediate networks obtained by (cherry or reticulated-cherry) reductions until reaching $N^S=I_4$, the reduction of $N$ with respect to the complete sequence $S=(3,1)(3,2)(1,2)(3,4)(2,4)$. For instance, the second network is $N^{(3,1)}$, the result of the reduction of (the reticulated-cherry) $(3,1)$ in $N$.}\label{fig:reduct}
\end{figure}

\section{Minimum reducible sequences}
\label{sec:min_red_seq}

As observed before, there may exist different complete reducible sequences for a given orchard network. Our goal in this section is to define a unique representative among all sequences giving the same network.

Let $(i,j)$, $(i',j')$ be two pairs of different integers. We say that $(i,j)\le (i',j')$ if $i<i'$ or $i=i'$ and $j\le j'$. If $(i,j)\le(i',j')$ and $(i,j)\neq(i',j')$, we simply write $(i,j)<(i',j')$.

Given two sequences of pairs of integers of the same length, $S=s_1\cdots s_k$ and $S'=s'_1\cdots s'_k$, we say that $S< S'$ if, for some $l\in[k]$ we have that $s_1=s_1',\dots,s_{l-1}=s'_{l-1}$ and $s_l<s'_{l}$.

It is easy to check that the relations just defined are total orders (on pairs and sequences of pairs of fixed length, respectively).

Given a non trivial orchard network $N$, consider the set $\RP(N)$ of reducible pairs of $N$. We define the \emph{minimum reducible pair} of $N$, $\MRP(N)$, as the minimum (with respect to the ordering just defined) pair in $\RP(N)$. Also, we denote by $\CRS(N)$ the set of complete reducible sequences of $N$, and we define the \emph{minimum complete reducible sequence} of $N$, $\MCRS(N)$, as the minimum (with respect to the ordering just defined) of $\CRS(N)$.  

Following the example of the two complete reducible sequences $S=(3,1)(3,2)(1,2)(3,4)(2,4)$ and 
$S'=(3, 1)(3, 4)(2, 3)(1, 3)(3, 4)$
for the orchard network $N$ depicted in Fig.~\ref{fig:reduct}, notice that since $(3,2)<(3,4)$ then $S<S'$. In fact, it can be checked that $\MCRS(N)=S$.

Notice that all the complete reducible sequences of a given orchard network have the same length, since this length is equal to $|X|+r-1$, where $r$ is the number of reticulations of $N$. 
We show that the two minimums just defined are related.

\begin{prop}\label{thm:mcrs_vs_mrp}
    Let $N$ be a non trivial orchard network. Then, the first pair in $\MCRS(N)$ is $\MRP(N)$.
\end{prop}

\begin{proof}
    Let $s$ be the first pair in $S=\MCRS(N)$ and $s'=\MRP(N)$. Obviously, $s\in\RP(N)$ and, from the definition of $\MRP(N)$, it follows that $s'\le s$.
    Due to \cite[Proposition~4.1]{erdos2019class}, the sequence with the single pair $s'$ can be extended to give a complete sequence $S'\in\CRS(N)$. Since the minimum complete sequence is $S$, we have that $S\le S'$, and hence $s\le s'$. Therefore, $s=s'$ and the result is proved.
\end{proof}

We define $\cS(X,r)$ as the set whose elements are the sequences $\MCRS(N)$ for every orchard network $N$ over $X$ with exactly $r$ reticulations.

\begin{thm}
\label{thm:bijection}
    There is a bijection between $\cS(X, r)$ and the set of orchard networks over $X$ with exactly $r$ reticulations.
\end{thm}
\begin{proof}
 The result follows from Theorem \ref{thm:same-crs-implies-iso} and the unicity of $\MCRS(N)$. 
\end{proof}


\section{Augmentation of networks}
\label{sec:augmentation}

In this section, we present an augmentation construction, which is the inverse of the reduction defined before, and show how we can determine the ARP of the obtained network from that of the original network.

Throughout this section we consider that $N$ is a network on $X\subseteq [n]$ and $(i,j)\in [n]\times [n]$ is a pair of integers with $i\neq j$ and $j\in X$.

We define the \emph{augmentation} of $(i,j)$ in $N$, denoted by $\aug{(i,j)}N$, as the result of:
\begin{itemize}
\item if $i \notin X$, create a new (leaf) node $i$, subdivide the arc ending in $j$ creating an elementary node $p_j$, and add the arc $p_ji$.  
\item if $i \in X$, subdivide both arcs ending in $i$ and $j$ creating elementary nodes $p_i$ and $p_j$, and add an arc $p_jp_i$. 
\end{itemize}

Similarly as in the reduction case, we shall define the \emph{augmentation} of an orchard network $N$ (which could be a trivial network $I_l$) with respect to a sequence $S=s_1\cdots s_k$ as $\prescript{s_1}{}( \cdots(^{s_{k-1}}(^{s_k}N)))$  and it will be denoted by $\aug SN$.

Notice that $(i,j)$ is a cherry in $\aug{(i,j)}N$, in symbols $\chi(i,j)=\tC$, when $i\notin X$, and $(i,j)$ is a reticulated cherry in $\aug{(i,j)}N$, in symbols $\chi(i,j)=\tR$, when $i\in X$. Then, the augmenting operation \emph{is} the inverse of the reduction operation, in the sense that $\aug{(i,j)}{\big(N^{(i,j)}\big)}\cong N$. This leads to present an alternative definition for orchard networks as those can be obtained by an augmentation of a trivial network $I_l$. 

Note also that if $N=\aug{S}{I_l}$ (for some $l\in[n]$), then necessarily the last pair in $S$ must be $(i,l)$ (for some $i\in[n]$). Hence, $l$ is determined by $S$ and can be omitted from $I_l$. Therefore, from now on we will simply write $N=\aug SI$.

We describe now how one can compute $\ARP(\aug{(i,j)}N)$ from $\ARP(N)$. That is, we show how the cherries and reticulated cherries of $N'=\aug{(i,j)}N$ can be found from the knowledge of those of $N$. Some remarks are due. 

\begin{enumerate}
    \item It is clear that the augmentation is a local operation; more precisely, a cherry (resp. reticulated cherry) in $N$ that is disjoint from $(i,j)$ keeps being a cherry (resp. reticulated cherry) in $N'$. 
    \item One only needs to check if the augmentation operation makes that some reducible pair disappears or changes its character (passes from cherry to reticulated cherry or viceversa), and if some new reducible pair appears. As for this last case, notice that the only reducible pair that can appear is $(i,j)$.
\end{enumerate}

Hence, we shall take any pair $(x,y)$ and decide if it is a reducible pair in $N'$ (that is, whether or not $(x,y)\in\RP(N')$) and, in such a case, if either $(x,y)^\tC\in\ARP(N')$ or $(x,y)^\tR\in\ARP(N')$ (equivalently, the value of $\chi_{N'}(x,y)$):
\begin{itemize}
  \item Case $\{x,y\}=\{i,j\}$:
  \begin{itemize}
    \item Case $i\notin X$: Both $(i,j)$ and $(j,i)$ are cherries in $N'$ and hence $(i,j)^\tC,(j,i)^\tC\in \ARP(N')$.
    \item Case $i\in X$: Now $(i,j)$ is a reticulated cherry and hence $(i,j)^\tR\in\ARP(N')$, but $(j,i)\notin\RP(N')$.
  \end{itemize}
\end{itemize}
  Note that from now on we can restrict ourselves to pairs $(x,y)$ in $\RP(N)$, since no other new pairs can appear.
\begin{itemize}
  \item Case $\{x,y\}\cap\{i,j\}=\emptyset$: From the local character of augmentation, $(x,y)\in\RP(N')$ and
  $\chi_{N'}(x,y)=\chi_N(x,y)$.
  \item Case $x=i, y\neq j$ (see Fig.~\ref{fig:casex_eq_i}): If $(x,y)$ is a cherry in $N$, say that $p$ is their common parent, then in $N'$ the arc $pi=px$ is split introducing a node which will be a reticulation; hence, $(x,y)$ is a reticulated cherry in $N'$. If $(x,y)$ is a reticulated cherry in $N$, then the parent of $y$ will no longer be a grandparent of $x=i$ in $N'$ (since the arc leading to $i$ is split in two). In brief, $(x,y)^\tR\in \ARP(N')$ if $(x,y)^\tC\in\ARP(N)$, and $(x,y)\notin \RP(N')$ otherwise.

\begin{figure}[!ht]%
\centering
\includegraphics[width=0.5\textwidth]{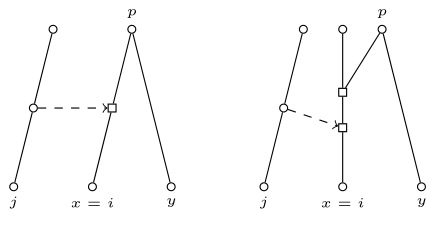}
\caption{Support picture for the case $x=i$ and $y\neq j$.  In the left, $(x,y)$ is a cherry in $N$. In the right, $(x,y)$ is a reticulated cherry in $N$.  The dashed arrow indicates the added arc to transform $N$ into $N'$ by the augmentation operation. Arcs whose tips are not explicitly drawn go from top to bottom.}\label{fig:casex_eq_i}
\end{figure}

  \item Case $x\neq j, y=i$ (see Fig.~\ref{fig:casey_eq_i}): The same argument as in the previous case gives that if $(x,y)$ is a cherry in $N$, then $(y,x)$ (notice the transposition) is a reticulated cherry in $N'$ (and hence $(x,y)\notin\RP(N')$).  Note that, if $(x,y)$ is a cherry of $N$, so is $(y,x)$, and hence the fact that $(y,x)^\tR$ belongs to $\ARP(N')$ will be covered by the application to the previous case applied to $(y,x)$. As before, if $(x,y)$ is a reticulated cherry of $N$, then it is no longer reducible in $N'$. 
  Therefore, in either case we have that $(x,y)\notin\RP(N')$.

\begin{figure}[!ht]%
\centering
\includegraphics[width=0.5\textwidth]{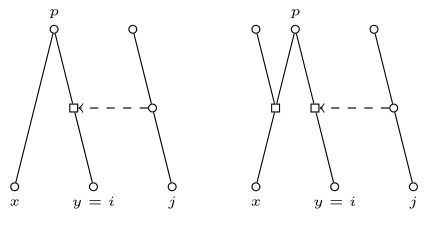}
\caption{Support picture for the case $x \neq j$ and $y=i$.}\label{fig:casey_eq_i}
\end{figure}

  \item Case $x\neq i, y=j$ (see Fig.~\ref{fig:casey_eq_j}): Let $py=pj$ be the arc leading to $y=j$; this arc is split in $N'$ by introducing a node that will be a tree node; this implies that $(x,y)$ will no longer be reducible in $N'$ and hence $(x,y)\notin \RP(N')$.

\begin{figure}[!ht]%
\centering
\includegraphics[width=0.5\textwidth]{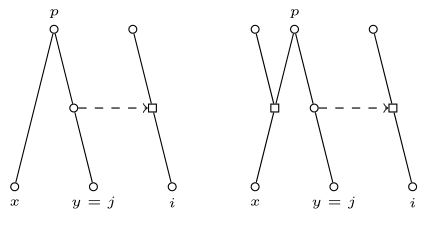}
\caption{Support picture for the case $x \neq i$ and $y= j$.} \label{fig:casey_eq_j}
\end{figure}

  \item Case $x=j, y\neq i$: The same argument as in the previous case, taking now the arc leading to $x=j$ implies that $(x,y)\notin \RP(N')$.
\end{itemize}

We can summarize these computations in the following result.

\begin{thm}\label{thm:augmentation_ARP}
    Let $N$ be an orchard network on $X\subseteq[n]$, and $(i,j)\in[n]\times[n]$ a pair with $i\neq j$ and $j\in X$.
    Consider
    the set of pairs $(x,y)\in[n]\times[n]$ such that one of the following conditions hold:
    \begin{enumerate}
        \item $(x,y)\in\RP(N)$, and $\{x,y\}\cap\{i,j\}=\emptyset$,
        \item $(x,y)\in\RP(N)$, $i=x$, $j\neq y$, and $\chi(x,y)=\tC$,
        \item $(x,y)=(i,j)$,
        \item $(x,y)=(j,i)$, and $i\notin X$.
    \end{enumerate}
    Annotate these pairs with the character $\chi'(x,y)$ given, in each case, by:
    \begin{enumerate}
        \item $\chi'(x,y)=\chi(x,y)$,
        \item $\chi'(x,y)=\tR$,
        \item $\chi'(x,y)=\tC$ if $i\notin X$, and $\chi'(x,y)=\tR$ if $i\in X$,
        \item $\chi'(x,y)=\tC$.
    \end{enumerate}
    Denote by $\aug{(i,j)}{\ARP(N)}$ the set of annotated pairs that is obtained by application of the procedure above. Then, $\aug{(i,j)}{\ARP(N)}=\ARP(\aug{(i,j)}N)$.
\end{thm}

As a result of the last theorem, in order to compute the set of annotated reducible pairs of the augmentation of a network, it is enough to traverse the annotated reducible pairs of the network. 

The next proposition shows that, given $S=\MCRS(N)$, it can be checked if $(i,j)S=\MCRS\big(\expair{i}{j}N\big)$ using only the information in $\ARP(N)$, without the need for knowing $N$ itself.

\begin{prop}
\label{prop:augmenting-condition}
 Let $S=\MCRS(N)$. Then, $(i,j)S=\MCRS\big(\aug{(i,j)}N\big)$ if and only if $(i,j)=\MRP\big(\aug{(i,j)}N\big)$.
\end{prop}

\begin{proof}
It is a direct consequence of Proposition~\ref{thm:mcrs_vs_mrp}.
\end{proof}


\section{Augmentation sequences and generation of orchard networks}
\label{sec:aug_seq_generation}

The goal of this section is to present an algorithm to generate the set of orchard networks over a set $[n]$ with exactly $r$ reticulations. Thanks to  Theorem \ref{thm:bijection}, this is equivalent to compute $\cS([n], r)$. Our strategy is to build these sequences starting with sequences of length one and, step by step, finding all possible pairs that can be prepended in order to get the sequences in $\cS([n], r)$.

Let $S=s_1\cdots s_k$ be a sequence of pairs of integers, say $s_t=(i_t,j_t)$ where $t=1,\dots,k$. We call the \emph{support} of $S$ the set $\Supp(S)=\{i_1,j_1,\dots,i_k,j_k\}$.
	For every $t=1,\dots,k$, we denote by $S_t$ the \emph{suffix} $s_t\cdots s_k$. 
	We say that a sequence $S$ as above is an \emph{augmentation sequence} if for each $t=1,\dots,k-1$, we have that $j_t\in \Supp(S_{t+1})$ and $i_t\neq j_t$. We remark that, although the formulation is not exactly the same, what we call augmentation sequences corresponds to \emph{cherry-picking sequences} in \cite[Definition~6]{Janssen2021}.
 
	It is clear that given an augmentation sequence $S$, we can consider the orchard network $N=\aug{S}I$, and also that $S$ will be a complete reducible sequence for $N$. From now on, all properties that can be defined for networks (taxa, number of reticulations, \dots) will be defined for augmentation sequences by applying them on the network that the sequence generates. For instance, we can define $\MRP(S):=\MRP(\aug{S}{I})$ and $\MCRS(S):=\MCRS(\aug SI)$.
	Note also that some of the properties can be found without having to construct the network itself.  For instance, the number of reticulations of $S=(i_1,j_1)\cdots(i_k,j_k)$, which by definition is the number of reticulations of $\aug SI$ can be found counting for how many indices $t=1,\dots,k-1$ we have that $i_t\in\Supp(S_{t+1})$. Also, using Theorem~\ref{thm:augmentation_ARP} recursively, we can compute $\ARP(S)$.

	We shall say that an augmentation sequence $S$ is a \emph{minimum} augmentation sequence if $S=\MCRS(N)$ for some network $N$. It is clear that it happens exactly when $S=\MCRS(S)$, and recall that $\MCRS(S)=\MCRS(\aug SI)$. This provides an alternative definition for $\cS(X,r)$ as the set of augmentation sequences that are stable under application of $\MCRS$, with support $X$ and with $r$ reticulations.

  These notations allow us to translate many properties that have been stated in terms of orchard networks into the language of sequences. For instance, Propositions~\ref{thm:mcrs_vs_mrp} and~\ref{prop:augmenting-condition} can be rewritten as follows.

\begin{prop}\label{thm:mcrs_vs_mrp_bis}
  Let $S$ be a minimum augmentation sequence.
  Then: 
  \begin{enumerate}
    \item The first pair in $S$ is $\MRP(S)$.
    \item Given a pair $(i,j)$ with $j\in\Supp(S)$, $(i,j)S$ is a minimum augmentation sequence if, and only if, $(i,j)=\MRP((i,j)S)$.
  \end{enumerate}
\end{prop}

We give now two results that characterize the suffixes of minimum augmentation sequences and, in particular, show that the last pair in such a sequence has a well determined form.

\begin{lem}\label{lem:minimal_suffix}
  Let $S=s_1\cdots s_k$ be a minimum augmentation sequence. Then, every suffix $S_t=s_t\cdots s_k$ ($t=1,\dots,k$) is a minimum augmentation sequence.
\end{lem}

\begin{proof}
    It is clear that $S_t\in\CRS(S_t)$. If there existed some $S_t'\in\CRS(S_t)$ with $S_t'<S_t$, then the concatenation $S'=s_1\cdots s_{t-1}S_t'$ would be strictly smaller than $S$ and also $S'\in\CRS(S)$, against the minimality of $S$.
\end{proof}

\begin{prop}\label{prop:last_pair_in}
  Let $(i,j)$ be the last pair in a minimum augmentation sequence $S\in\cS([n],r)$. Then, $j=n$.
\end{prop}

\begin{proof}    
    Write $\MCRS(N)$ as $S=(i_1,j_1)\cdots(i_k,j_k)$, where $(i_k,j_k)=(i,j)$ and assume that $j\neq n$.

    Suppose first that $i=n$ and $j<n$. From Lemma~\ref{lem:minimal_suffix}, $(n,j)$ is a minimum complete reduction sequence (of the cherry $\aug{(n,j)}I$), but in this case $(j,n)$ is also a complete reduction sequence, and $(j,n)<(n,j)$, leading to a contradiction.

    Now, we can assume that $i,j<n$.
    Let $t<k$ be such that $(i_t,j_t)$ is the last pair where one of its entries is $n$ (it exists because $S\in\cS([n],r)$, and hence $n\in\Supp(S)$).
Now, $S_{t+1}=(i_{t+1},j_{t+1})\cdots(i_k,j_k)$ is a minimum augmentation sequence thanks to Lemma~\ref{lem:minimal_suffix}. Since $n$ does not belong to $\Supp(S_{t+1})$, but does belong to the support of $S_t=(i_t,j_t)S_{t+1}$, we have that $i_t=n$ and $j_t<n$.
Moreover, $(n,j_t)$ is a cherry in $N'=\aug{S_{t}}I$. Thanks again to Lemma~\ref{lem:minimal_suffix}, $S_t=(n,j_t)S_{t+1}$ is a minimum augmentation sequence and, 
thanks to Proposition~\ref{thm:mcrs_vs_mrp}, $(n,j_t)=\MRP(N')$. However, since $(n,j_t)$ is a cherry of $N'$, then $(j_t,n)$ is also a cherry, and $(j_t,n)<(n,j_t)$, leading to a contradiction.
\end{proof}

For every $m\in\{1,\ldots,n-1\}$ we define the set 
$$\cS_m([n],r)=\{\Tilde{S} \mid \Tilde{S}\in\cS([n],r)\text{ and $\Tilde{S}$ ends in $(m,n)$}\}.$$
A direct consequence of Proposition~\ref{prop:last_pair_in} is the following result, that states that 
the computation of $\cS([n],r)$ is reduced to the computation of the subsets $\cS_m([n],r)$.

\begin{prop}
\label{prop:baseAlg2}
    $\cS([n],r)=\bigsqcup_{m=1}^{n-1} \cS_m ([n],r)$.
\end{prop}

We know, from Proposition~\ref{prop:last_pair_in}, the form of the last pair in a minimum augmentation sequence. It is also clear that any such pair $(m,n)$ is a minimum augmentation sequence.
Our next result shows how minimum augmentation sequences can be extended by prepending pairs of integers in order to generate other minimum augmentation sequences.

\begin{thm}\label{thm:augmentation_sequences}
    Let $S'=(i,j)S$ be an augmentation sequence. Then, $S'$ is a minimum augmentation sequence if, and only if, $S$ is a minimum augmentation sequence and $(i,j)=\MRP((i,j)S)$. In such a case, say that $S\in \cS(X,r)$ and $S'\in \cS(X',r')$. If $i\in\Supp(X)$, then $r'=r+1$ and $X'=X$; otherwise, $r'=r$ and $X'=X\cup\{i\}$.
\end{thm}

\begin{proof}
  The non-trivial parts of the statement follow from Proposition~\ref{thm:mcrs_vs_mrp_bis}.
\end{proof}

Using these results, is easy to give a procedure that generates all the orchard networks over a set $[n]$ of taxa and with a given number $r$ of reticulations. Indeed, it is enough to generate, for each positive integer $m<n$, the set $\cS_m([n],r)$, and the latter can be generated as follows:
\begin{enumerate}
    \item Start with the sequence $S=(m,n)$, of length $1$, with support $X=\{m,n\}$, and whose set of  annotated reducible pairs is $\ARP=\ARP(S)=\{(m,n)^\tC,(n,m)^\tC\}$. 
    \item Recursively, given a sequence of pairs $S$, with support $X$, and given also the set $\ARP=\ARP(S)$, find all possible pairs $(i,j)$ such that $(i,j)=\MRP((i,j)S)$. For each such $(i,j)$, consider the extended sequence $S'=(i,j)S$ with support $X'=S\cup\{i\}$ and with set of annotated reducible pairs $\ARP'=\aug{s}\ARP$.
    \item If $X=[n]$ and the length of the obtained sequence is $|S|=n+r-1$, then 
    yield the sequence $S$. 
\end{enumerate}

\begin{thm}\label{thm:correctness_generation}
  The set of sequences yielded by the procedure above is $\cS([n],r)$.
\end{thm}

\begin{proof}
  Let $S$ be a sequence yielded by the procedure.
  The condition that $S$ has support $[n]$ and has $r$ reticulations is guaranteed by the condition in step 3 of the procedure. The condition that $S$ is a minimum augmentation sequence follows by applying recursively Theorem~\ref{thm:augmentation_sequences}, thanks to the condition in step 2, and with the starting condition in step 1 being justified by Proposition~\ref{prop:baseAlg2}.

  Conversely, if $S=s_1\cdots s_k\in \cS([n],r)$, then  $s_k=(m,n)$ for some $m$ (thanks to Proposition~\ref{prop:baseAlg2}), and it will be considered in step 1. At each step,  considering the suffix $S_t=s_t\cdots s_k$ in step 2, the pair $s_{t-1}$ will fulfill the conditions (thanks to Theorem~\ref{thm:augmentation_sequences}), and hence $S_{t-1}=s_{t-1}s_t\cdots s_k$ will be considered in the next iteration. Finally, in step 3, the sequence $S$ will be yielded.
\end{proof}

Some remarks are due:
\begin{enumerate}
    \item The set $\aug{s}\ARP$ in step 2 can be computed using Theorem~\ref{thm:augmentation_ARP}, and it can be done in linear time with respect to the length of $S$. Also, if the pairs in $\ARP$ are stored increasingly ordered with respect to the lexicographic ordering, then the computation of $\aug{s}\ARP$ can be performed so that $\aug{s}\ARP$ keeps being ordered and, in particular, its minimum element can be found in constant time.
    \item Another advantage of storing the pairs in $\ARP$ ordered is that, in order to determine if $(i,j)=\MRP((i,j)S)$, one does not need to compute the whole set $\aug{(i,j)}{\ARP}$. Indeed, in the process of building $\aug{(i,j)}{\ARP}$, at most three pairs in $\ARP$ can disappear, and hence one only needs to take the first four elements in $\ARP$, decide which of them belong to $\aug{(i,j)}{\ARP}$, and test if $(i,j)$ is smaller than each of those.
    \item Given a minimum augmentation sequence $S$, it is possible that it can not be extended to another minimum augmentation sequence $(i,j)S$ if we want to keep the number of reticulations. For instance, if we consider the sequence $S=(1,2)(2,4)\in\cS_2(\{1,2,4\},0)$, the only possible extensions that keep the number of reticulations are obtained by prepending one of the pairs $(3,1)$, $(3,2)$ or $(3,4)$; however, none of these sequences is minimum, as can be easily checked in each case.
    \item The search of extensions can be pruned. For instance, if at a given stage, the sequence $S$ has $r$ reticulations, the only pairs $(i,j)$ that have to be considered are those with $i\notin \Supp(S)$, since otherwise the number of reticulations would be greater than $r$.
    \item Also in the case that we are adding a cherry (that is, when $i\notin \Supp(S)$) we can restrict ourselves to the case that $i<j$, since otherwise $(j,i)$ would be a reducible pair in $\aug{(i,j)S}I$, and since $(j,i)<(i,j)$, it is impossible that $(i,j)=\MRP((i,j)S)$.
    \item The algorithm can be easily modified, so that instead of generating all the sequences with exactly $r$ reticulations, it generates all sequences with \emph{at most} $r$ reticulations.
\end{enumerate}


\section{Generation of tree-child networks}
\label{sec:generation_tc}

A network is \emph{tree-child} if every node that is not a leaf has a child that is a tree node~\cite{Cardona2009}. For brevity, we shall simply say that each interior node has a \emph{tree child}.

The same procedure we have described to generate orchard networks can be adapted to generate all tree-child networks over $[n]$, by adding some conditions to ensure that the generated sequences correspond to tree-child networks.

First, we need to decide when the reductions and augmentations defined in the previous sections produce tree-child networks.

We start with the following result, adapted from \cite[Lemma~4.1]{Bordewich2016}, that states that reductions of tree-child networks are tree-child networks.
\begin{lem}
    \label{lem:tc-is-orch}
    Let $N$ be a tree-child network. Then, $N$ is an orchard network and, if $(i,j)\in\RP(N)$, then $N^{(i,j)}$ is also a tree-child network.
\end{lem}

In order to decide whether or not an augmentation of a tree-child network is tree-child, we need to introduce new terminology.
Let $N$ be a network over $X\subseteq[n]$. Then, we define the \emph{state} $\sigma_N(i)$ of $i\in[n]$ as follows:
\begin{itemize}
\item if $i\notin X$, then $\sigma_N(i)=\tN$;
\item otherwise, if the parent of $i$ is a reticulation, then $\sigma_N(i)=\tP$;
\item otherwise, if the sibling of $i$ is a reticulation, then $\sigma_N(i)=\tS$;
\item otherwise, $\sigma_N(i)=\tT$.
\end{itemize}

If the network is clear from the context we shall simply write $\sigma(i)$ instead of $\sigma_N(i)$. Then, $\sigma_N : [n] \to \{\tN, \tP, \tS, \tT\}$ is a mapping that gives the state of each $i\in[n]$ in $N$. We also define the \emph{state of a network} $N$ as $\sigma(N)=(\sigma_N(1), \ldots, \sigma_N(n))$.
Finally, if $S$ is an augmentation sequence, we shall denote $\sigma(S)=\sigma(\aug{S}{I})$.

The following result gives the conditions under which an augmentation produces a tree-child network.

\begin{thm}\label{thm:reduction_tc}
  Let $N$ be an orchard network. Then,
$\expair{i}{j}{N}$ is a tree-child network if, and only if, $N$ is tree-child and $\sigma_N(i)\in\{\tN, \tT\}$.
\end{thm}
\begin{proof}
  Let $N'=\expair{i}{j}{N}$. From Lemma~\ref{lem:tc-is-orch} we know that if $N'$ is tree-child, $N$ is also tree-child.

  Now, suppose that $\sigma_N(i)=\tP$. Then, $i$ is a leaf in $N$ and its parent $p_i$ is a reticulation. When applying the augmentation $(i,j)$, the arc $p_ii$ is split introducing a new node $v$ that shall become a reticulation. Then, in $N'$, the only child of $p_i$ is $v$, which is a reticulation. Therefore, $N'$ is not tree-child. 
  
  Similarly, suppose that $\sigma_N(i)=\tS$. Then, $i$ is a leaf in $N$, its parent $p_i$ is a tree node and its sibling $s_i$ is a reticulation. Again, in the process of applying the augmentation $(i,j)$, the arc $p_ii$ is subdivided introducing a new reticulation $v$. Thus, the children of $p_i$ in $N'$ are $v$ and $s_i$, both reticulations, so $N'$ is not tree-child, against the hypothesis. Therefore, $\sigma_N(i)\notin\{\tP,\tS\}$, which is equivalent to $\sigma_N(i)\in\{\tN,\tT\}$.

  Conversely, assume that $N$ is tree-child. Due to the local nature of the augmentation processes, the condition that each node (other than a leaf) in $N'$ has a tree child needs only to be tested for the nodes that are adjacent to the leaves involved in the augmentation.

  First, assume that $\sigma_N(i)=\tN$, and let $p_j$ be the parent of $j$ in $N$. The augmentation process creates a tree node $v$ in $N'$ with children $i,j$ and parent $p_j$. Now, $p_j$ keeps having a tree child (the node $v$), and the new internal node $v$ has both children that are tree nodes (the leaves $i$  and $j$). Hence, the condition of being tree-child is preserved.

  Second, assume that $\sigma_N(i)=\tT$, which implies that $i\in X$  and hence the augmentation process creates two elementary nodes: $u$ (a tree node) between $j$ and its parent $p_j$, and $v$ (a reticulation) between $i$ and its parent $p_i$. Also, since $\sigma_N(i)=\tT$, we have that the sibling $s_i$ of $i$ in $N$ (that is, the child of $p_i$ in $N$ different from $i$) is a tree node. In $N'$, $p_j$ has $u$ as a tree child, $p_i$ has $s_i$, $u$ has $j$, and $v$ has $i$. Hence, the condition of being tree-child is preserved. 
\end{proof}

We describe now how to compute $\sigma(\expair{i}{j}{N})$ from $\sigma(N)$. For simplicity, we write $N'=\expair{i}{j}{N}$, $\sigma=\sigma_N$ and $\sigma'=\sigma_{N'}$, and we will restrict to the cases of interest that $\sigma(i)\in\{\tN,\tT\}$.

\begin{itemize}
\item Case $\sigma(i)=\tN$. In this case, $(i,j)$ is a cherry in $N'$, and therefore $\sigma'(i)=\sigma'(j)=\tT$. From the local behavior of the augmentation, for any other leaf $l$ in $N$, its parent (and its sibling, in case it has one) remains the same. Therefore, we conclude that $\sigma'(i)=\sigma'(j)=\tT$ and $\sigma'(l)=\sigma(l)$ for all $l\in[n]\setminus\{i,j\}$.

\item Case $\sigma(i)=\tT$. In this case, $(i,j)$ is a reticulated cherry in $N'$, hence $\sigma'(i)=\tP$ and $\sigma'(j)=\tS$. Thanks again to the local behaviour of the augmentation, the state of a leaf $l$ in $N'$ can only differ from its state in $N$ if its parent or sibling change from being a tree node to a reticulation (or viceversa). Hence, only siblings of $i$ and $j$ have to be taken into consideration.
If $j$ was the sibling of another leaf $l$ in $N$, then $l$ would still have a sibling in $N'$ that is a tree node (namely, the parent of $j$ in $N'$) and hence the state of $l$ would not change. If $i$ was the sibling of another leaf $l$ in $N$, which can be written as $(i,l)^\tC\in\ARP(N)$, then $l$ would change from having a sibling that is a tree node (the leaf $i$) to having a sibling that is a reticulation (the parent of $i$ in $N'$). Hence $\sigma(l)=\tT$ but $\sigma'(l)=\tS$.
\end{itemize}

We can summarize these computations in the following result.

\begin{thm}
  \label{thm:comp_sigma}
Let $N$ be a tree-child network over $X\subseteq[n]$ with state function $\sigma$. Let $j\in X$ and $i\in[n]$, $i\neq j$, with $\sigma(i)\in\{\tN,\tT\}$. 
Consider  the function $\aug{(i,j)}\sigma:[n]\to \{\tN,\tP,\tS,\tT\}$ defined as follows:
\begin{itemize}
  \item If $\sigma(i)=\tN$,
    \begin{itemize}
    \item $\aug{(i,j)}{\sigma}(i)=\aug{(i,j)}{\sigma}(j)=\tT$,
    \item $\aug{(i,j)}{\sigma}(l)=\sigma(l)$ for all $l\in[n]\setminus\{i,j\}$.
    \end{itemize} 
  \item If $\sigma(i)=\tT$,
    \begin{itemize}
    \item $\aug{(i,j)}{\sigma}(i)=\tP$,
    \item $\aug{(i,j)}{\sigma}(j)=\tS$,
    \item $\aug{(i,j)}{\sigma}(l)=
    \begin{cases}
      \tS, & \text{if } (i,l)^\tC\in\ARP(N), \\
      \sigma(l), & \text{ otherwise}
    \end{cases}$ for all $l\in[n]\setminus\{i,j\}$.
  \end{itemize}
\end{itemize}
Then $\expair{i}{j}{N}$ is a tree-child network over $X\cup\{i\}$ with state function $\aug{(i,j)}{\sigma}$.
\end{thm}

We denote by $\cT(X, r)$ the subset of $\cS(X, r)$ formed by sequences $S$  such that ${}^SI$ is a tree-child network. 
Thanks to Lemma~\ref{lem:tc-is-orch}, the set $\cT(X, r)$ is in bijection with the set of tree-child networks over the set of taxa $X$ and with $r$ reticulations. Also, notice that $\cT([n],r)=\bigsqcup_{m=1}^{n-1} \cT_m ([n],r)$, where $\cT_m([n],r)$ denotes, for every $m\in\{1,\ldots,n-1\}$, the subset of $\cS_m([n], r)$ formed by sequences $S$ such that ${}^SI$ is tree-child. Finally, notice that in the case of tree-child networks, the number of reticulations $r$ is bounded by $n-1$~\cite[Proposition~1]{Cardona2009}.

Then, we can modify the procedure that generates all orchard networks over $[n]$ with $r$ reticulations to generate all tree-child networks over $[n]$ with $r$ reticulations, provided that $r<n$.

Indeed, for every positive integer $m<n$, we can generate the sets $\cT_m([n], r)$ as follows:
\begin{enumerate}
\item Start with the sequence $S=(m,n)$, of length $1$, with support $X=\{m,n\}$, with set of annotated reducible pairs $\ARP=\ARP(S)=\{(m,n)^\tC,(n,m)^\tC\}$ and with state $\sigma=\sigma(S)$ whose entries are all $\tN$ except the $m$-th and $n$-th entry which are $\tT$. 

\item Recursively, given a sequence of pairs $S$ (and assuming that $\aug{S}I$ is tree-child), with support $X$, and given also the set $\ARP=\ARP(S)$ and the state $\sigma=\sigma(S)$, find all possible pairs $(i,j)$ such that $\sigma(i)\in\{\tN,\tT\}$ and  $(i,j)=\MRP((i,j)S)$. For each such $(i,j)$, consider the extended sequence $S'=(i,j)S$ with support $X'=S\cup\{i\}$, with set of annotated reducible pairs $\ARP'=\aug{s}\ARP$ and with state $\sigma'=\aug{s}{\sigma}$.

\item If $X=[n]$ and the length of the obtained sequence is $|S|=n+r-1$, then 
    yield the sequence $S$. 
\end{enumerate}

\begin{thm}\label{thm:correctness_generation_tc}
  The set of sequences yielded by the procedure above is $\cT([n],r)$.
\end{thm}

\begin{proof}
  The result follows using the same reasoning as in Theorem~\ref{thm:correctness_generation}, using now Theorem~\ref{thm:reduction_tc} to ensure that the yielded networks are tree-child.
\end{proof}

Some remarks follow:
\begin{enumerate}
  \item The state function $\sigma'$ in step 2. can be computed using Theorem~\ref{thm:comp_sigma}, and notice that the information in $\ARP$ is also needed.
\item As in the case of orchard networks, the procedure can be adapted to yield all tree-child networks over $[n]$ with at most $r$ reticulations. In particular, since tree-child networks over $[n]$ have at most $n-1$ reticulations \cite[Proposition~1]{Cardona2009}, we can generate all of them.

\item The procedure given for generating tree-child networks can be adapted to generate all stack-free~\cite{semple2018phylogenetic} orchard networks, simply checking if $\sigma(i)\neq\tP$ instead of checking if $\sigma(i)\in\{\tN,\tT\}$.
\end{enumerate}


\section{Computational experiments}
\label{sec:computational}

The procedure to generate orchard and tree-child networks described in this paper has been implemented in C. Source files, documentation and examples are available in the repository~\url{https://github.com/gerardet46/OrchardGenerator}. Notice that the output of the implementation are complete reducible sequences, given as strings, and that they can be used as input to build and manipulate networks using the Python package PhyloNetworks~\cite{PhyloNetwork}.

There are some interesting details to comment. First, as we said, the set $\ARP$ is kept ordered, and the cherries $(i,j)$ with $i>j$ are ignored. Taking this into account, notice that if $N$ is an orchard network on $X\subseteq[n]$, it holds that $|\ARP(N)|\le\frac{2}{3}|X|\le \frac{2}{3}n$ (and there is always an orchard network such that the equality holds). Therefore, the set $\ARP$ can be implemented as an static array, which is much faster than a dynamic one. Also, given a sequence $S=s_1\ldots s_k$, we store the set $\{\ARP(s_k),\ARP(s_{k-1}s_k),\ldots,\ARP(S)\}$ for faster access when trying different candidate extensions.

Notice also that the only data needed to store the networks is $X$, $S$ and $\ARP$ (and $\sigma$ for tree-child networks), but there is no need to store the network $N$ itself. Also, the operations involved in the algorithm are very simple, so they could be easily implemented in C, optimizing the performance.

We have also implemented a random orchard network generator that follows the same lines of the procedure to generate all the networks, but choosing a random pair at each step in order to produce a sequence, instead of trying all the candidates. Notice however that this generator does not generate networks uniformly. Indeed, even at the first step,  the number of $\MCRS$ ending in $(n-1,n)$ is greater than the number of those ending in $(1,n)$.

Finally, the algorithm can be parallelized, considering a partition of suffixes and creating a process for each subset of suffixes, which generates all sequences ending in a suffix from the corresponding subset.

Using this implementation, we have computed the number of orchard networks for small number of leaves and reticulations, shown in Table~\ref{tab:number_orchard}. 
As for the generation of tree-child networks, it is  worth to mention the speed of the computation compared to previously implemented methods. Indeed, Table~\ref{tab:speed} shows the time of execution for the generation of all tree-child networks with $n=5$ leaves using the implementations of the results in this paper compared to those in \cite{cardona2019generation} and \cite{cardona2020counting}.

\begin{table}[!ht]
\begin{tabular}{c|rrrrr}
\toprule
$r \backslash n$ & 2   & 3       & 4          & 5             & 6               \\ \midrule
0                  & 1   & 3       & 15         & 105           & 945             \\
1                  & 2   & 21      & 228        & 2\,805          & 39\,330           \\
2                  & 4   & 132     & 2\,832       & 57\,150         & 1\,185\,300         \\
3                  & 8   & 804     & 32\,880      & 1\,054\,200       & 31\,481\,280        \\
4                  & 16  & 4\,848    & 370\,320     & 18\,520\,320      & 783\,492\,840       \\
5                  & 32  & 29\,136   & 4\,107\,648    & 316\,583\,280     & 18\,766\,151\,280     \\
6                  & 64  & 174\,912  & 45\,197\,952   & 5\,323\,207\,200    & 438\,647\,126\,400    \\
7                  & 128 & 1\,049\,664 & 495\,183\,360  & 88\,589\,126\,400   & 10\,087\,314\,094\,080  \\
8                  & 256 & 6\,298\,368 & 5\,412\,422\,400 & 1\,464\,596\,709\,120  & 229\,383\,137\,571\,840 \\
\midrule
Time                 & 0.00s   & 0.02s & 5.99s & 1\,693.39s & 470\,828.27s  \\ \bottomrule
\end{tabular}
\caption{Number of orchard networks with $n$ leaves and $r$ reticulations, for $2\le n\le 6$ and $0\le r\le 8$, together with the total time used to compute these numbers, for each value of $n$.}
\label{tab:number_orchard}
\end{table}

\begin{table}[]

    \begin{tabular}{l|r}
    \toprule
    Implementation in Python from \cite{cardona2020counting}     &  9m19.249s \\
    Implementation in Python from \cite{cardona2019generation}     &  7m23.162s \\
    Implementation in C of the current paper     &  0m00.056s \\
    \bottomrule     
    \end{tabular}
        \caption{Time needed for the generation of all tree-child networks with $n=5$ leaves using different implementations.}
    \label{tab:speed}
\end{table}


\section{Conclusions}
\label{sec:conclusions}

Phylogenetic networks model evolutionary relationships among organisms and overcome the limitations of using phylogenetic trees by allowing the representation of reticulate processes. 

In this paper, we have considered the problem of the efficient and injective generation of all orchard and tree-child networks (with a given number of leaves and reticulations), two special classes of phylogenetic networks with biological relevance \cite{kong2022classes}.  Our method is based on considering sequences of pairs of integers that characterize those networks \cite{Janssen2021} and finding a subset of those (called minimum complete reducible sequences) that characterize the networks injectively.

To this end, we have first shown that such a sequence must end in a pair $(m,n)$, where $n$ is the desired number of leaves and $m<n$, and that we can iteratively extend the sequences by prepending new pairs to generate the sequences that encode the networks. This method is efficient since there is no need to construct the network itself in order to check if the candidate sequence effectively corresponds to an orchard (or tree-child) network.

The implementation of the algorithms described in the paper allows a fast generation of the sequences (and implicitly of the networks). For example, our implementation is capable of generating all orchard networks with $4$ leaves and at most $8$ reticulations, of which there are about $6$ billions of them, in approximately $6$ seconds.\footnote{Computation performed on a computer with two processors Intel\textsuperscript{\tiny\textregistered} Xeon\textsuperscript{\tiny\textregistered} E5-2690 (3.00GHz), providing 40 CPUs.} For  tree-child networks, we have shown that our method is much faster than other methods previously published and implemented. 

There are some natural questions that arise as a possible future work, mainly in the direction of extending our results to the generation of other classes of phylogenetic networks. One of the possible generalizations is getting rid of the binary condition and generating semi-binary and non-binary (orchard and tree-child) networks. In this sense, the results in \cite{Janssen2021} could be applied, considering more possible annotations of pairs, in order to cover the six different reductions that this paper considers.
Another direction could be trying to use other topological conditions on the networks to be generated. For instance, and as we have commented at the end of Section~\ref{sec:generation_tc}, only a small change in our method is needed in order to generate   stack-free orchard networks. Another potential subclass of networks where our methods could apply is the class of normal networks \cite{willson2010properties}, which is a subclass of tree-child networks where shortcuts are not allowed (that is, if two nodes are linked by an arc, then they cannot be connected by another path).

\section*{Acknowledgment} 
GC and JCP were supported by Grant PID2021-126114NB-C44 funded by MCIN/AEI/10.13039/\\
501100011033 and by ``ERDF A way of making Europe''.

\bibliographystyle{alpha}
\bibliography{biblio}

\end{document}